\def\draft{0}
\def\llncs{0}

\ifnum\draft=0

\fi

\ifnum\llncs=0
\documentclass[11pt]{article}
\usepackage{fullpage}
\else
\documentclass[orivec,runningheads,envcountsame,envcountsect]{llncs}
\usepackage{amssymb,amsmath,cite,url}
\fi

\usepackage{bm,bbm}
\newcommand{\remove}[1]{}
\newcommand{\ignore}[1]{}
\ifnum\llncs=0
\usepackage{amsthm}
\else
\fi

\ifnum\draft=1
\def\ShowAuthNotes{1}
\else
\def\ShowAuthNotes{0}
\fi

\usepackage{mathtools}
\usepackage{amsmath,amssymb}
\usepackage{url}
\usepackage{cite}
\usepackage{tabularray,tabularx}

\usepackage{lipsum} 
\usepackage[]{mdframed}
\usepackage{floatrow}
\usepackage{multicol}

\usepackage[dvipsnames]{xcolor}
\definecolor{DarkBlue}{RGB}{0,0,150}
\definecolor{llg}{gray}{0.95}
\definecolor{lg}{gray}{0.85}

\usepackage[colorlinks,linkcolor=DarkBlue,citecolor=DarkBlue]{hyperref}

\ifnum\llncs=0

\newtheorem{theorem}{Theorem}[section]
\newtheorem{proposition}[theorem]{Proposition}
\newtheorem{definition}[theorem]{Definition}

\newtheorem{lemma}[theorem]{Lemma}

\newtheorem{remark}{Remark}
\newtheorem{corollary}[theorem]{Corollary}

\else

\spnewtheorem{prop}{Property}{\bfseries}{\itshape}
\spnewtheorem{fact}{Fact}{\bfseries}{\itshape}
\spnewtheorem{subclaim}{Claim}[theorem]{\bfseries}{\itshape}
\spnewtheorem{tlclaim}[theorem]{Claim}{\bfseries}{\itshape}
\spnewtheorem{assumption}{Assumption}{\bfseries}{\itshape}
\fi

\newcommand{\secparam}{\lambda}
\newcommand{\secp}{\secparam}

\ifnum\llncs=1
\def\pfend{\hfill\qedsymbol}
\else
\def\pfend{}
\fi

\def\cK{{\cal K}}

\def\cM{{\cal M}}

\def\cS{{\cal S}}

\def\bbE{{\mathbb E}}

\def\bbI{{\mathbb I}}

\def\bbN{{\mathbb N}}

\def\binset{\{0,1\}}

\def\q2{\lfloor q/2 \rceil}

\newcommand{\abs}[1]{\left\vert {#1} \right\vert}
\newcommand{\norm}[1]{\left\| {#1} \right\|}

\newcommand{\adv}{\mathrm{Adv}}

\newcommand{\mx}[1]{\mathbf{{#1}}}

\newcommand{\zo}{\{0,1\}}

\newcommand{\Ex}{\mathop{\bbE}}

\ifnum\ShowAuthNotes=1
\newcommand{\authnote}[3]{\textcolor{#3}{[{\footnotesize {\bf #1:} { {#2}}}]}}
\else
\newcommand{\authnote}[3]{}
\fi

\newcommand{\absnewline}{\ifnum\llncs=1 \\ \fi}

\def\mgp[#1]{\mx{G}_{#1}}

\def\mgip[#1]{\mx{G}_{#1}^{-1}}

\def\mcit[#1]{\mci[#1]^T}
\def\mci[#1]{\mx{C}_{#1}}

\newcounter{hybridcount}
\newcounter{prevhybridcount}
\newcounter{nexthybridcount}

\newcommand{\hrho}{{\hat{\rho}}}

\newcommand{\ket}[1]{|{#1}\rangle}
\newcommand{\bra}[1]{\langle{#1}|}
\newcommand{\ketbra}[1]{\ket{{#1}}\bra{{#1}}}

\newcommand{\tr}{\mathop{\mathrm{Tr}}}

\def\zon{{\binset^{n}}}

\def\adv{{\mathcal{A}}}

\newcommand{\haardist}{\mu}

\newcommand{\haarint}[2]{\int \ketbra{{#1}}^{\otimes {#2}} d\haardist({#1})}

\newcommand{\statedist}{\cS}

\newcommand{\keygen}{\mathsf{KeyGen}}
\newcommand{\stategen}{\mathsf{StateGen}}

\usepackage{tikz}
\usetikzlibrary{quantikz2}
\usepackage{subcaption}

\title{State-Based Classical Shadows}

\author{Zvika Brakerski\thanks{Weizmann Institute of Science. \texttt{Email:} \{zvika.brakerski,  nir.magrafta\}@weizmann.ac.il. Supported by the Horizon Europe Research and Innovation Program via ERC Project ACQUA (Grant 101087742).} \and Nir Magrafta\addtocounter{footnote}{-1}\footnotemark \and Tomer Solomon\thanks{Tel Aviv University. \texttt{Email:} tomersolomon@mail.tau.ac.il. Supported in part by AFOSR FA9550-23-1-0387, AFOSR FA9550-23-1-0312, and ISF 2338/23.}}
\date{}

\begin{document}

\maketitle

\begin{abstract}	
	Classical Shadow Tomography (Huang, Kueng and Preskill, Nature Physics 2020) is a method for creating a classical snapshot of an unknown quantum state, which can later be used to predict the value of an a-priori unknown observable on that state. In the short time since their introduction, classical shadows received a lot of attention from the physics, quantum information, and quantum computing (including cryptography) communities. In particular there has been a major effort focused on improving the efficiency, and in particular depth, of generating the classical snapshot.

    Existing constructions rely on a distribution of unitaries as a central building block, and research is devoted to simplifying this family as much as possible. We diverge from this paradigm and show that suitable distributions over \emph{states} can be used as the building block instead. Concretely, we create the snapshot by entangling the unknown input state with an independently prepared auxiliary state, and measuring the resulting entangled state. This state-based approach allows us to consider a building block with arguably weaker properties that has not been studied so far in the context of classical shadows. Notably, our cryptographically-inspired analysis shows that for \emph{efficiently computable} observables, it suffices to use \emph{pseudorandom} families of states. To the best of our knowledge, \emph{computational} classical shadow tomography was not considered in the literature prior to our work.

     Finally, in terms of efficiency, the online part of our method (i.e.\ the part that depends on the input) is simply performing a measurement in the Bell basis, which can be done in constant depth using elementary gates.

\end{abstract}

\tableofcontents

\section{Introduction}

Quantum states are elusive objects. A quantum state over $n$ qubits is fully characterized by its \emph{density matrix} which is a $2^n \times 2^n$ complex valued, positive semidefinite, unit-trace matrix. This means that such a state may require an infinite amount of information to fully specify, and even a finite-precision characterization (say up to a fixed $\ell_1$ distance in matrix norm) requires an exponential amount of information. Indeed, the task of fully characterizing a quantum state (i.e.\ approximating its density matrix) to a finite precision, known as \emph{state tomography}, requires an exponential number of copies of the quantum state \cite{ODonell2016,Haah2017,Song201710QubitEA}. Nevertheless, in terms of information content, given $n$ qubits containing the state, it is only possible to recover $n$ bits of information \cite{holevo1973bounds}.

 We may therefore seek relaxed notions of characterization that require fewer resources. One such relaxation, known as \emph{shadow tomography} \cite{Aaronson2018}
 does not aim to recover the density matrix. Instead, it allows to use a bounded number of copies of the state in order to estimate the outcomes of an exponential number of potential measurements of the input state. Moreover, the number of required copies only scales polylogarithmically with the number of potential measurements to be estimated. On the other hand, shadow tomography requires a quantum algorithm that has quantum access to the copies of the state, and the descriptions of the measurements to be estimated. Its running time in general can be exponential in the number of copies of the state.

Recently, Huang, Kueng and Preskill \cite{huang2020predicting}, proposed the notion of \emph{classical shadows}. They also consider the task of estimating a large number of measurement outcomes on an unknown quantum state. However, in their outline each copy of the quantum state is measured independently to recover some classical information. Then, given an observable, its value can be estimated using the classical information obtained from all of the individual copies. Specifically, the classical information recovered from each copy of the state can be used to create a ``classical snapshot'' $\hrho$ of the density matrix $\rho$ of the input state. The snapshots converge to the actual density matrix in expectation, and furthermore, the value of every intended measurement on the input state can be estimated by properly combining its values on a sufficient number of snapshots. The quality of this estimate depends on a property of the measurement to be estimated which is called its ``shadow norm''. They showed that for certain classes of measurements, it is possible to achieve good estimates with a polylogarithmic scaling of the number of copies, similarly to shadow tomography.

Classical shadows are quite appealing since they are quite parsimonious in terms of quantum resources. Indeed, the copies of the quantum input state need to be measured in a prescribed manner, and from that point and on, only classical information needs to be stored and processed. Importantly, the processing of the quantum state is indifferent with respect to the specific observable to be estimated. Classical shadows therefore found applications in the characterization of physical systems \cite{zhang2021experimental}, in quantum complexity theory \cite{anshu2024survey}, in quantum information theory \cite{kunjummen2023shadow}, and even in quantum cryptography \cite{cryptoeprint:2023/1620}. 

The usefulness of classical shadows generated a large volume of work on improving its efficiency, e.g.~\cite{akhtar2023scalable,hu2023classical,bertoni2024shallow,hearth2024efficient,schuster2024random}. A central parameter for optimization is the circuit depth of generating the snapshots. Indeed, quantum computers quickly accumulate noise throughout the computation so as a general rule, shallower quantum circuits should be easier to implement than deeper ones.\footnote{Of course there is significant dependence on the actual topology and connectivity of the circuit, but the depth is a reasonable proxy for the noise accumulation if we do not wish to commit to a specific quantum architecture.} For the snapshot generation, therefore, one wishes to minimize the complexity (specifically depth) of the unitary operator that is applied on the input state.

In the context of this paper, we are concerned with the setting of estimating the values of \emph{general observables} on the input state, which is one of the two main settings considered in \cite{huang2020predicting}.\footnote{The other main setting considered there is that of \emph{local} observables. We believe that our work may have some bearing in that setting as well, but more exploration is required to establish whether this is the case.} An observable in this context refers to estimating the expected value of a projective measurement performed on the state. More explicitly, the observable is characterized by a Hermitian matrix $O$, and the value of the observable on the input state is the trace $\tr(O\rho)$. One can verify that this represents the expected value of a process consisting of performing a measurement of the state $\rho$ in some basis, followed by assigning a real value to each basis element and outputting this value.%

To generate classical snapshots, \cite{huang2020predicting} consider a ``sufficiently random'' distribution over unitary operators. Then they sample a unitary from this distribution, apply it on the input state and measure in the standard (computational) basis. The outcome of this measurement, together with the identity of the specific unitary that had been sampled, are post-processed classically to produce the classical snapshot. In the setting of a general observables described above, the proposed distribution is simply the Haar-random distribution over unitaries. Whereas a completely random unitary is extremely inefficient to generate, it turns out that the analysis goes through even when the random unitary is replaced with a weaker object: a unitary $3$-design. This is a distribution over unitaries whose moments, up to the third moment, are identical to those of a random unitary. Fortunately, such objects can be efficiently sampled and implemented in linear depth.\footnote{A knowledgeable reader may be aware that random Clifford circuits are known to be unitary $3$-designs.} Recently, very significant progress has been made when Schuster, Haferkamp and Huang \cite{schuster2024random} presented new constructions of approximate unitary designs with extremely low depth. Indeed, they showed that it is possible to generate approximate unitary $3$-designs using logarithmic depth circuits. Their notion of approximation is \emph{relative}: When viewed as a quantum channel, their unitaries are within $(1\pm \epsilon)$ factor of a random unitary channel. They show that this notion of approximation translates into a tolerable penalty in the performance of the classical shadows.
If we seek to approximate unitary designs, then it appears that this is the best that we can hope for. Indeed, \cite{schuster2024random} present a lower bound showing the tightness of their results with respect to the depth complexity of approximate unitary designs.

The approach we take in this work is a cryptographic one. We view the classical shadows problem as a \emph{pseudorandomness problem}, and ask whether there exist relaxed structures with a lower cost that recover similar performance to the random unitary solution of \cite{huang2020predicting}.\footnote{The random Clifford solution of \cite{huang2020predicting} and the approximate design solution of \cite{schuster2024random} can already be viewed in this perspective as information-theoretic pseudorandom distributions that fool a certain process.}
Specifically, we present a different framework that breaks away from the blueprint of ``apply a unitary and measure''. Our framework does not formally require unitary designs and it allows to minimize the \emph{online depth complexity} of the snapshot generation process. Namely, to minimize the depth of the operations that are applied to the input state. 
Our new blueprint relies on an arguably weaker building block: approximate state-designs (instead of approximate unitary designs). We then show that whereas prior solutions require relative-approximate designs, meaningful results can already be achieved even under a weaker \emph{additive notion}. We do this using a cryptographic approach by presenting a suitable distinguisher for the state design construction. Finally, we observe that the complexity of our distinguisher is roughly as the complexity of the observable to be estimated. This allows us to introduce \emph{computational} pseudorandomness to the setting of classical shadows for efficient observables. Additional details are provided below.

\subsection{Our Results: State-Based Classical Shadows}

Our framework is rather simple. Instead of considering a distribution over unitaries, we consider a distribution over quantum states. Our tomography process, given an input state $\rho$, is to sample a state $\zeta$ from this distribution, and measure the pair $\rho \otimes \zeta$ in the so-called Bell basis. The Bell basis is a basis for $2$-qubit systems, in which all basis elements are maximally entangled, namely correlated in the strongest quantum sense. We measure the first qubit of $\rho$ together with the first qubit of $\zeta$ and so on. This effectively entangles the states $\rho$ and $\zeta$ and produces measurement outcomes that depend on the correlation between the two states. See Figure~\ref{fig:design-circuit} for a circuit diagram of our snapshot generation process, where the Bell basis measurement is presented explicitly using elementary gates.

Before discussing the correctness of our method, and the way we select the distribution of states, let us discuss the complexity of our process. We note that the state $\zeta$ can be generated in a way that is completely independent of the input $\rho$. This means that $\zeta$ can be prepared ahead of time. We may therefore consider the online quantum complexity of the snapshot generation, which in this case is simply the Bell basis measurement which can be implemented by $3$ layers of elementary gates, as shown in the figure. Indeed, a Bell basis measurement is equivalent to the generation of EPR pairs, one of the most elementary operations in quantum computing. Going back to the motivation above, our snapshot generation process causes minimal disturbance to the input state $\rho$ which would allow to characterize it most accurately. One should still be concerned why it is justified to put less weight on the complexity of generating the auxiliary state $\zeta$. First, we point out that in the contexts we consider below, the cost of generating ``useful'' $\zeta$ state is bounded from above by the process of applying an approximate unitary design, and potentially may be lower. So at the very least our method preserves the total complexity of previous methods while reducing the online complexity. Second, since the process of generating $\zeta$ can be done offline, prior to receiving the input $\rho$, and is completely under our control, we may employ stronger noise reduction and fault tolerance techniques in the generation of $\zeta$, including the use of error correction and post-selection, similarly to the techniques used in so-called ``magic-state factories'' e.g.~\cite{o2017quantum}. Indeed, our $\zeta$ plays a very similar role to that of a magic state (magic states \cite{bravyi2005universal} are a method for using the Clifford family of gates to compute gates that are not in the family, using preprocessed auxiliary states).

Let us now see how to analyze state-based classical shadows. We start by calculating the probability of measuring classical values $z,x$ in the measurements in Figure~\ref{fig:design-circuit}, while keeping $\zeta, \rho$ fixed. We see that the probability of $z,x$ is proportional to the correlation between $\rho$ and a state we denote by $\zeta^*_{z,x}$. This state is the complex conjugate of the state $\zeta$, and the state is also conjugated by a Pauli matrix $X^x Z^z$. To the unfamiliar reader it suffices to say that the action of $X^x$ is just the XOR of the state with the fixed string $x$, whereas $Z^z$ is the same operation, but applied in the Hadamard basis. We may therefore consider (as a thought experiment) the quantum channel that applies the measurement, and based on $z,x$ outputs $\zeta^*_{z,x}$. This channel is fully characterized by $\zeta$, and we can in-principle compute its inverse given the distribution of states $\zeta$.

Alas, for $\zeta$ sampled from a general distribution, this calculation is non-trivial and we are unable to present an analytic solution. We do believe it is quite an interesting problem to present a general analysis, or show that this task is intractable, but leave it as a question for future work. We can, however, analyze the performance of our estimator when $\zeta$ is a \emph{state $3$-design}, and recover the same performance as \cite{huang2020predicting}. Similarly to unitary designs discussed above, state designs are state ensembles whose first moments (properly defined) are identical to those of a Haar random quantum state. Similarly to unitary designs, it is possible to generate state designs efficiently using a polynomial amount of randomness. In fact, state designs are implied by unitary designs trivially, by applying the unitary design on a fixed state, say $\ket{0}$. Therefore, our requirement is no-stronger, and potentially weaker than that of previous works, even regardless of our complexity advantage.

We can analyze the performance of our method also when $\zeta$ is sampled from an \emph{approximate} state design with relative approximation factor $\epsilon$. This is defined analogously to approximate unitary designs mentioned above, but somewhat more simply since states are simpler than channels. We recover the results of \cite{schuster2024random} in this case. We go beyond the analysis in prior work and consider the setting of state designs with \emph{additive} approximation $\epsilon$. This means that the moments of the distribution are $\epsilon$-close to those of a random state, up to matrix $\ell_1$ distance $\epsilon$. We show that this notion of approximation, too, allows to recover meaningful classical shadows, so long as $\epsilon$ is sufficiently small: $\epsilon \ll 2^{-2n}$. We use this result as a segue to our analysis of \emph{computational} notions of classical shadows, described further below. We stress that our result for the additive setting is in principle \emph{incomparable} to the \cite{schuster2024random}-inspired result that uses relative distance. Indeed, similarly to the trade-off shown in \cite{schuster2024random,Brand_o_2016}, we observe that any additive $\epsilon$-approximate state $3$-design is also a relative $2^{3n}\epsilon$-approximate state $3$-design. However, our additive results only require $\epsilon \ll 2^{-2n}$ to be applicable (at least in some parameter regime), a regime in which the relative approximation is trivial. Furthermore, whereas our relative analysis (following \cite{schuster2024random}) works natively only for \emph{positive} observables, our additive analysis works for any observable. Whereas it is possible to map any observable $O$ into a positive version $O'$ so that the expectation of $O$ can be derived from that of $O'$, but this conversion may be computationally laborious.

The specific calculation aside, in order to convey the qualitative properties of our approach, let us try to ex post facto recover it from the principles of \emph{quantum teleportation}. Consider a setting where we wish to implement the \cite{huang2020predicting} approach and apply some unitary $U$ to a state $\rho$, and measure the outcome.
Instead, let us attempt to use teleportation as follows. We generate $n$ EPR pairs, apply $U$ to one end of the set of EPR pairs, and teleport $\rho$ into the other end.  Let $x,z$ be the values recovered in the teleportation process, and recall that in teleportation we apply a $X^x Z^z$ correction to the destination register. A simple calculation would show that right before the $X^x Z^z$ correction, the destination register is in the state $U^T X^x Z^z \rho$, i.e.\ we effectively applied $U'_{z,x}=U^T X^x Z^z$ to $\rho$.  Alas, we seem to have failed in our task. We were unable to apply our desired unitary $U$ and incurred a penalty of $X^x Z^z$ for values $z,x$ that cannot be predicted. Our first insight is that we do not need to apply a prescribed unitary in order to obtain classical shadows, as in the prevalent blueprint. Rather, we only need to be able to recover sufficient classical information that specifies the description of the unitary that had been applied. Furthermore, we observe that since our goal is to use ``sufficiently random'' unitaries $U$, the ``Pauli penalty'' of $X^x Z^z$ does not interfere with the statistical properties of the applied unitary. Indeed, the above protocol already works, when $U$ is a unitary design (or approximate design), and achieves constant online depth. We take a step further and consider running the above method, and stopping after applying $U$ to one end of the EPR pairs,  but before measuring $\rho$ and the EPR pair. We notice that at this point we are simply entangling $\rho$ with a pre-processed state that has been prepared in a prescribed manner using $U$. This gives rise to the generalization presented in Figure~\ref{fig:design-circuit} where we encapsulate the preparation of $\zeta$, and only require that it has appropriate statistical properties,  thus avoiding the requirement of the ``stronger'' unitary design primitive. This simplification also results in a potential reduction in the number of qubits used, since no EPR pairs or teleportation is actually needed,  as we just collect statistics about $\rho$ rather than teleporting it.

Indeed, the only other work we are aware of that contains a related insight is \cite{mcginley2023shadow}. There, classical shadows are generated by appending $\ket{0}$ ancilla qubits to $\rho$, and applying a grand unitary on both registers together, thus entangling them, and measuring the outcome. However, this work also suggests to use a computationally labor intensive unitary on the pair of registers, which is in contrast to our motivation to reduce the computational overhead on $\rho$.

\paragraph{The Computational Setting.} In our analysis above, we explained that using approximate state designs, we may obtain useful classical shadows. This brings about the natural question of whether it is possible to use states that are \emph{computationally indistinguishable} from being a state design. This is motivated by the literature on constructing pseudorandom quantum states \cite{JLS18}.

At first, the blueprint seems quite simple. Let us consider states that are indistinguishable from quantum state designs. We call them ``state pseudo-designs''. We wish to show that substituting approximate designs for pseudo-designs would still result in adequate shadows, namely that they can properly be used to predict the expected values of observables on the input state $\rho$.

In order to push this argument through, we would need to argue that if the pseudo-design leads to a significant skew in the estimated value of the observable, then this discrepancy can lead to \emph{an efficient distinguisher} between the pseudo-design and an actual state design. This a-priori seems unclear. The calculation of the estimated value is not necessarily an efficient process, and furthermore the computed value is a real number and one would need to find a way to convert it into a binary output for a distinguisher.

To address this challenge, we revisit the analysis of the statistical setting. The analysis is done by considering the random variable obtained from generating a single snapshot and evaluating the observable on this snapshot. Then, the expected value and variance of this random variable are properly bounded, and from these bounds it is possible to derive a guarantee about the closeness of the aggregated snapshot to the actual value of the observable on the input state.

We notice that whereas the calculation of the expected value and variance may be quite inefficient classically, there exists an efficient quantum circuit that computes these values. This circuit only needs access to $3$ copies of the state $\zeta$ and black-box access to evaluate the complex-conjugate of the observable in question, in addition to some other elementary gates. Recall that we mentioned above that our snapshot generation produces a correlation between $\rho$ and $\zeta^*_{z,x}$ which is related to the complex conjugate of $\zeta$. Therefore, if we wish to compute the expectation / variance using $\zeta$, we require the complex conjugate of the observable. Indeed, the computational complexity of an observable is identical to that of its complex conjugate. We also show that it is possible to translate the output of the (conjugate) observable into a binary guess on whether the input state is random or not.

It follows that if our pseudo design is ``sufficiently indistinguishable'' from an actual state design with respect to a set of distinguishers of size $T$, then the pseudo design can be used to estimate the values of all observables of complexity $T'$ where $T' = \Omega(T)$. Indeed, as in the statistical setting, the required indistinguishability is quite strong, distinguishers of size $T$ are required to have an advantage $\ll 2^{-2n}$ in order to obtain a meaningful result. We note that there exist constructions of quantum pseudorandom states that can be converted into pseudo-designs with such advantage, these are the so-called \emph{scalable} pseudorandom state constructions \cite{brakerski2020scalable,lu2023quantum}. However, we are unable to ascertain advantageous properties of using these specific constructions compared to the state of the art information theoretic constructions. Nevertheless, we view this path of investigation as quite intriguing. In particular since it suffices to construct pseudo-designs against limited classes of distinguishers in order to achieve meaningful results.

We note that to the best of our knowledge, computational versions of classical shadows, or of shadow tomography, have not been studied in the literature prior to our work and we view our work as opening a new research direction in the study of classical shadows. In addition to the above question of coming up with better pseudo-designs, one may also try to recover a computational version of relative approximate designs which could lead to potentially stronger results.

Finally, we notice that our results in the computational setting imply previously unknown limitations on constructions of pseudorandom quantum states. In particular, it can be shown quite straightforwardly that if $\zeta$ is a \emph{real valued} quantum state, namely all of its coordinates (as a $2^n$ dimensional state vector) are real, rather than complex, then it cannot be used for classical shadows. This implies, via our result, that such constructions cannot serve as pseudo designs with a parameter of $2^{-2n}$ or better. Indeed, real valued constructions are known \cite{BS19,GB23,GMW23} and indeed their analysis only achieves a parameter of $2^{-n}$. Our result therefore shows that it is not possible to achieve the ``scalability'' property while remaining real-valued.

\paragraph{Other Future Directions.} As we mentioned above, we are only able to analyze our new framework when instantiating the state $\zeta$ as an approximate design or a pseudo design. It may be plausible that other distributions of state with ``sufficient tomographic capacity'' may also yield meaningful results for classical shadows. It is possible that being close to a state design is not required at all and some milder property which is more easily achievable would be sufficient.

\section{Preliminaries}

\subsection{General Quantum Notations}

We use standard notation for quantum states and operators. If $\zeta$ is a unit vector in an $n$-qubit ($2^n$ dimensional) Hilbert space, then we let $\ket{\zeta}$ denote the quantum state corresponding to this vector.

We fix a universal set of quantum gates, for convenience we assume that this set is closed under complex conjugation. We say a circuit $C$ is of size $T$ if it consists of $T$ gates.

We consider ensembles of quantum states of the form $\{\ket{\psi_k}\}_{k \sim \cK}$, where $\cK$ is a distribution over classical strings, such that each string $k$ in the support of $\cK$ corresponds to a quantum state $\ket{\psi_k}$.

If $A,B$ are operators then $A \preceq B$ is notation for $B-A$ being positive semidefinite. If $v$ is a real-valued vector then $A^v = \bigotimes_j A^{v_j}$.

\subsection{Haar Measure}
We denote the Haar measure over $n$-qubit states by $\mu$. We use the following property of the Haar measure:
\begin{proposition}[see e.g.~\cite{gross2015partial} Lemma 6,7]
\label{prop:haar_two}
Let $A,B$ be matrices of size $2^n \times 2^n$. Then,
\begin{equation*}
        2^n (2^n+1) \mathrm{Tr}_2 \left( \left( \haarint{\phi}{2} \right) A \otimes B \right) = \tr(B)A + BA ~,
\end{equation*}
and thus
\begin{equation*}
    2^n (2^n+1) \tr \left( \left( \haarint{\phi}{2} \right) A \otimes B \right) = \tr(A) \tr(B) + \tr(AB) ~.
\end{equation*}
\end{proposition}

\subsection{Observables}
\begin{definition}[Observable]
An observable $O$ over $n$ qubits is an Hermitian matrix of size $2^n \times 2^n$. Recall that therefore $O$ can be decomposed as $O = \sum_i \alpha_i \Pi_i$, where $\Pi_i$ are orthogonal projector, and $\alpha_i$ are the real-valued eigenvalues of $O$, where for $i \not = j$, $\alpha_i \not = \alpha_j$ and $\Pi_i \Pi_j = 0$.
\end{definition}

\begin{definition}[Observable Measurement]
Let $O = \sum_i \alpha_i \Pi_i$ be an observable. A measurement according to $O$ is a quantum algorithm that given a quantum state $\rho$, returns $\alpha_i$ with probability $\tr[\Pi_i \rho]$. We say the observable has complexity $T$ if it can be implemented by a quantum circuit of size $T$.
\end{definition}
In circuit diagrams, we denote the application of a measurement according to observable $O$ as:
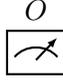
\begin{align*}
	\begin{quantikz}
		\meter{O}
	\end{quantikz}
\end{align*}

\subsection{Approximate State Designs}

A state $t$-design is an ensemble of states for which the $t$-th moment equals the $t$-th moment of a state taken from the Haar measure. Formally,

\begin{definition}[$t$-State Design]
    Let $t, \lambda\in \mathbb{N}$ and let $\cK$ be a distribution over $\binset^\secp$. An ensemble of states $\{ \ket{\psi_k} \}_{k \sim \cK}$ is a \emph{state $t$-design} if $$
    \haarint{\phi}{t}
    =
    \Ex_{k\leftarrow \cK} \left[\ketbra{\psi_k}^{\otimes t}\right]
    $$
\end{definition}

We will consider two definitions of approximate state $t$-designs.

\begin{definition}[Additive $\epsilon$-Approximate State $t$-Design]
\label{def:additive-state-design}
    Let $t, \lambda\in \mathbb{N}$ and let $\cK$ be a distribution over $\binset^\secp$. An ensemble of states $\{ \ket{\psi_k} \}_{k \sim \cK}$ is an \emph{additive} $\epsilon$-approximate state $t$-design if $$
    \norm{\haarint{\phi}{t}
    -
    \Ex_{k\leftarrow \cK} \left[\ketbra{\psi_k}^{\otimes t}\right]
    }_1 \le \epsilon
    $$
\end{definition}

\begin{definition}[Relative $\epsilon$-Approximate State $t$-Design]
    Let $t, \lambda\in \mathbb{N}$ and let $\cK$ be a distribution over $\binset^\secp$. An ensemble of states $\{ \ket{\psi_k} \}_{k \sim \cK}$ is a \emph{relative} $\epsilon$-approximate state $t$-design if $$
    (1-\epsilon)\haarint{\phi}{t}
    \preceq
    \Ex_{k\leftarrow \cK} \left[\ketbra{\psi_k}^{\otimes t}\right]
    \preceq
    (1+\epsilon)\haarint{\phi}{t}
    $$
\end{definition}

We also define the computational analogue of Definition~\ref{def:additive-state-design}:
\begin{definition}[State $t$-Pseudo-Design Generator]\label{def:pseudodesign}
	We say that a pair of polynomial-time quantum algorithms $(\keygen, \stategen)$ is an $(T,\epsilon)$-State $t$-Pseudo-Design Generator if the following holds:
	\begin{itemize}
		\item \textbf{Key Generation.}
		For all $n \in \bbN$, $\keygen(1^n)$ outputs a classical key $k$.
		
		\item \textbf{State Generation.}
		Given a key $k$ in the support of $\keygen(1^n)$, the algorithm $\stategen(1^n, k)$ outputs a $n$-qubit pure quantum state $\ket{\psi_k}$.
		
		\item \textbf{Pseudorandomness.}
		For any non-uniform quantum algorithm $A = \{ A_n \}_{n \in \bbN}$ (with quantum advice), where for any large enough $n$ the size of $A_n$ is at most $T(n)$, it holds that:
		$$
		\abs{\Pr[A_n \big( D_1 \big) = 1] -
			\Pr[A_n \big( D_2 \big) = 1]} \leq \epsilon(n) \enspace ,
		$$
		where the distributions $D_1, D_2$ are defined as follows.
		\begin{itemize}
			\item $D_1 : $
			Sample $k \gets \keygen(1^n)$, perform $t$ independent executions of $\stategen(1^n, k)$, and output $\ket{\psi_k}^{\otimes t}$.
			
			\item $D_2 : $
			Sample a random $n$-qubit quantum state $\ket{\psi} \gets \mu_n$, and output $t$ copies of it: $\ket{\psi}^{\otimes t}$.
		\end{itemize}
	\end{itemize}
\end{definition}

\section{State-Based Classical Shadows}
\begin{figure}
    \centering  
        \begin{mdframed}
        \centering
        \begin{quantikz}
            \lstick{{$R_1$ $\,\,\,\, \rho$}} & \qwbundle{n} &  \ctrl{1} & \gate{H^{\otimes n}} & \meter{z} \\
            \lstick{{$R_2$ $\,\,\,\,\zeta$}} & \qwbundle{n}  & \targ{} & & \meter{x}
        \end{quantikz}  
        \begin{enumerate}
        \item Denote the input register $R_1$.
            \item In register $R_2$, generate an auxiliary state $\zeta$ from a distribution of states $\statedist$.
            \item Apply $n$ CNOT gates qubit by qubit, with $R_1$ qubits as controls and $R_2$ qubits as targets.
            \item Apply $H^{\otimes n}$ on $R_1$.
            \item Measure $R_1$ and $R_2$ in the standard basis and denote the results by $z$ and $x$ respectively.
        \end{enumerate} 
    \end{mdframed}

    \caption{Description of a single experiment in the shadow tomography procedure.}
    \label{fig:design-circuit}
\end{figure}

We begin with presenting our new framework for generating a \emph{classical shadow} of an unknown quantum state. We consider a system containing two $n$-qubit registers, $R_1$ and $R_2$. The unknown input state $\rho$ is provided in $R_1$, and the register $R_2$ contains an auxiliary state $\ket{\zeta}$ with a known classical description \footnote{This is possible for example if the state is generated by a pseud-design generator, where the key $k$ if first generated. Note the full classical description the unit vector corresponding to $\ket{\zeta}$ has exponential length in $n$.}. We then run the quantum algorithm depicted in Figure~\ref{fig:design-circuit} to obtain classical results $x,z$. Let $\ket{\zeta_{x,z}} \coloneqq X^x Z^z \ket{\zeta}$. The final classical result, which following \cite{huang2020predicting} we denote as the \emph{classical shadow} $\hat{\rho}$\footnote{Sometimes this is also called the \emph{classical snapshot}.}, is the classical matrix\footnote{Note that when $\ket{\zeta_k}$ is determined by some classical key $k$, we could store $k,x,z$ and only calculate $(2^n+1)\ketbra{\zeta^*_{x,z}} - I$ when needed for observable estimation.}

$$
\hat{\rho} = (2^n+1)\ketbra{\zeta^*_{x,z}} - I \enspace .
$$
Each such classical shadow $\hat{\rho}$ induces a corresponding estimation $\tr(O \hat{\rho})$\footnote{Note that given a classical description of the matrix $\hat{\rho}$ we can calculate analytically the value $\tr(O \hat{\rho})$.} of $\tr(O\rho)$. We then aggregate multiple copies of $\hat{\rho}$ to estimate $\tr(O\rho)$ using median of means \cite{nemirovski1983medianmeans,JERRUM1986169}. Let $\hat{o}(L,K)$ be the median of $K$ means, each on $L$ estimates $\tr(O \hat{\rho})$. Clearly, we get different estimation guarantees depending on the state ensemble $\statedist$ the state $\ket{\zeta}$ is sampled from. In the following, for an observable $O$, denote as $O_0$ the traceless part of $O$, i.e. $O_0 =  O - \frac{\tr(O)}{2^n}\bbI$. We have the following theorems to bound the estimation error\footnote{These theorems are similar in spirit to \cite[Theorem 1]{huang2020predicting} and \cite[Theorem 6]{schuster2024random}, but with approximate state design instead of unitaries.}:

\begin{theorem}[Relative $\epsilon$-Approximate State Design Guarantees]
\label{theorem:relative-guarantees}
Let $O$ be a $2^n \times 2^n$ \emph{positive} Hermitian matrix, for $n \geq 2$, and let $\gamma, \delta \in (0,1)$ be accuracy parameters. Then, assuming that a \emph{relative} $\epsilon$-approximate $3$-state design is used in the shadow tomography process described in figure \ref{fig:design-circuit}, a collection of $LK$ independent classical shadows allows approximating:
\begin{align}
\Pr[\abs{\hat{o}(L,K) - \tr(O\rho)} \leq \gamma + 2\epsilon\tr(O)] \geq 1 - \delta \enspace,
\end{align}
for
$$
K = 2\log(2/\delta) \enspace \text{and} \enspace L = \frac{34}{\gamma^2} \left( 3\tr(O_0^2) + 10\epsilon \tr(O)^2 \right)
$$
\end{theorem}

\begin{theorem}[Additive $\epsilon$-Approximate State Design Guarantees]
\label{theorem:additive-guarantees}
Let $O$ be a $2^n \times 2^n$ Hermitian matrix, and let $\gamma, \delta \in (0,1)$ be accuracy parameters. Then, assuming that an \emph{additive} $\epsilon$-approximate $3$-state design is used in the shadow tomography process described in Figure~\ref{fig:design-circuit}, a collection of $LK$ independent classical shadows allows approximating:
\begin{align}
\Pr[\abs{\hat{o}(L,K) - \tr(O\rho)} \leq \gamma + (2^n+1)\varepsilon \norm{O}_{\infty} ] \geq 1 - \delta \enspace,
\end{align}
for
$$
K = 2\log(2/\delta) \enspace \text{and} \enspace L = \frac{34}{\gamma^2} \left(  3 \tr( O_0^2 ) +  3\epsilon\norm{O}^2_\infty(2^n+1)^2 \right)
$$
\end{theorem}

\begin{theorem}[$(T,\epsilon)$-State $t$-Pseudo-Design Guarantees]
\label{theorem:computational-guarantees}
Let $O$ be observable over $n$ qubits with complexity $t$. Then there exists a constant $c$ such that for $T = c \cdot \max(t, n)$, if $(T, \epsilon)$-State $3$-pseudo-design Generator is used in the shadow tomography process described in figure \ref{fig:design-circuit}, a collection of $LK$ independent classical shadows allows approximating:
\begin{align}
\Pr[\abs{\hat{o}(L,K) - \tr(O\rho)} \leq \gamma + 2(2^n+1)\varepsilon \norm{O}_{\infty} ] \geq 1 - \delta \enspace,
\end{align}
for
$$
K = 2\log(2/\delta) \enspace \text{and} \enspace L = \frac{34}{\gamma^2} \left(3 \tr( O_0^2 ) +  6\epsilon\norm{O}^2_\infty(2^n+1)^2 \right)
$$
\end{theorem}

The proofs to Theorems \ref{theorem:relative-guarantees},\ref{theorem:additive-guarantees} and \ref{theorem:computational-guarantees} are given in Section~\ref{section:proofs-main-theorem} by plugging in the bounds on the bias and variance from Theorems \ref{theorem:relative_snapshot}, \ref{theorem:addative_snapshot} and \ref{theorem:computational_snapshot}.

\begin{remark}[Comparing Relative and Additive Approximations]
To illustrate the difference in applicability regimes of our theorems, and in particular the distinction between additive and relative approximation, we consider a family of states, and let $\epsilon$ (respectively $\epsilon'$) be its additive (respectively relative) approximation factor as an approximate $3$-design (note that $\epsilon \lesssim \epsilon' \lesssim 2^{3n} \epsilon$, see Appendix~\ref{apx:addrelstate}). If we use this approximate design for state-based classical shadows, which theorem should we use for the analysis? 

We perform the comparison by plugging the parameters into Theorems \ref{theorem:relative-guarantees}, \ref{theorem:additive-guarantees}. We note that one should consider both the bias in the expectation and the rate of convergence towards this bias, i.e.\ the number of samples required to be sufficiently close to the bias. Indeed, the bias is fully determined by $\epsilon, \epsilon'$ (and the properties of the observable), so if one method ``wins'' in this category then the other cannot match the approximation, no matter how many samples are used (this is in contrast to unbiased estimators such as \cite{huang2020predicting} where any approximation is achievable given sufficiently many samples). The rate of convergence determines the number of required samples, so it is possible to ``win'' in this category while ``losing'' in the bias.

Concretely, we notice that in the relative case (Theorem~\ref{theorem:relative-guarantees}), up to constants, the bias is $\epsilon' \tr(O)$ (we only consider positive observables) and the convergence rate is governed by $\epsilon \tr(O)^2$ (where we ignore terms that are the same between the two cases). Likewise, in the additive case the bias is governed by $\epsilon \cdot 2^n \norm{O}_{\infty}$ and the convergence rate is governed by $\epsilon \cdot 2^{2n}  \norm{O}_{\infty}^2$. Let $\alpha = \tr(O)/(2^n \norm{O}_{\infty})$. By definition $2^{-n} \le \alpha \le 1$, depending on the observable. 

In terms of bias, therefore, we need to compare $\epsilon' \alpha$ to $\epsilon$. Indeed, whenever $\epsilon \le 2^{-n} \epsilon'$, the additive analysis wins in the bias, and for some operators it wins (up to a constant) across the entire regime. In terms of the convergence, we need to compare $\epsilon' \alpha^2$ to $\epsilon$. Now the additive analysis wins on all observables only if $\epsilon \le 2^{-2n} \epsilon'$. So we see that there is a regime where the additive analysis wins across all observables in the bias, but loses over all observables in the convergence.

It follows that for concrete constructions of approximate designs, one should map out $\epsilon, \epsilon'$ and decide on which analysis to use accordingly.

\end{remark}

\begin{remark}[Note on Exact Snapshot Reconstruction]
Our Shadow Tomography procedure can be divided into two steps:
\begin{enumerate}
    \item Quantum step to obtain a classical snapshot from input state $\rho$ and auxiliary state $\ket{\zeta}$. This step results in measuring $z,x$, and defining a classical snapshot $\ketbra{\zeta_{x,z}}$.
    \item Classical step to reconstruct a classical shadow $\hat{\rho}$ from $\ketbra{\zeta_{x,z}}$. Specifically, we define
    \begin{align}
    \hat{\rho} = (2^n+1)\ketbra{\zeta^*_{x,z}} - I \enspace .
    \end{align}
\end{enumerate}
More generally, for any distribution $\statedist$ over pure states, we can define the map from $\rho$ to the expectation of $\ketbra{\zeta_{x,z}^*}$ (over sampling $\zeta$ from $\statedist$). By Proposition~\ref{prop:probability}, this can be shown to be:
\begin{align}
\cM_\statedist (\rho) &\coloneqq \Ex_{{\zeta} \leftarrow \statedist} \left[ \sum_{x,z \in \zon} \Pr \Big[x,z \Big| {\zeta} \Big] \ketbra{\zeta_{x,z}^*} \right] 
\\ &= \frac{1}{2^n} \sum_{x,z \in \zon} \mathrm{Tr}_{R_1} \left( \Ex_{{\zeta} \leftarrow \statedist} \left[  {\ketbra{\zeta_{x,z}^*}^{\otimes 2}}_{R_1,R_2} \right] (\rho_{R_1} \otimes I_{R_2})  \right) ~.
\end{align}
Suppose the channel $\cM_\statedist$ is invertible, and let $\cM_\statedist^{-1}$ be its inverse mapping. Note $\cM_\statedist^{-1}$ is a linear map, though not necessarily a quantum channel. Assuming we know how to compute the inverse mapping $\cM_\statedist^{-1}$, we can then define the classical shadow to be
\begin{align}
    \hat{\rho} \coloneqq \cM_\statedist^{-1} \left( \ketbra{\zeta_{x,z}^*} \right) ~.
\end{align}
In this case, the expectation of $\cM_\statedist^{-1}(\ketbra{\zeta_{x,z}^*})$, over sampling $\zeta$ from $\statedist$ and the measurement outcomes $x,z$, equals $\rho$. However, finding such $\cM_\statedist^{-1}$ for an arbitrary distribution $\statedist$ might be hard. Instead, we can consider a close distribution $\widetilde{S}$ of states, and consider $\cM_{\widetilde{\statedist}}^{-1}$. In our case, we take $\widetilde{S}$ to be the Haar random distribution over pure states. In this case, $\cM_{\widetilde{\statedist}}$ is the depolarizing channel, so its inverse mapping is
\begin{align}
  \cM_{\widetilde{\statedist}}^{-1}(A) = (2^n+1)A- \tr(A)I  
\end{align}
The intuition is that if the distributions $S$ and $\widetilde{S}$ are close, then we can apply $\cM_{\widetilde{\statedist}}^{-1}$ instead of $\cM_{\statedist}^{-1}$. This intuition is proved in our theorems for the case $\statedist$ forms an approximate (or pseudo) state design, and $\widetilde{\statedist}$ is the Haar random distribution over pure states.
\end{remark}

\section{Analysis}
In this section we bound the bias of the expectation of estimating an observable $O$, as well as its variance, for both types of approximations. Toward that end we have the following propositions:
\begin{proposition}[Bias Bound]
\label{prop:bias-bound}
Let $\statedist$ be a distribution over pure states, and let $O$ be an observable. For a quantum state $\rho$, let $\hat{\rho}$ be the random variable of the classical shadow with respect to state distribution $\statedist$. Then:
\begin{align}
\label{eq:prop_exp}
\abs{\Ex_{\zeta \gets \statedist} \left[\tr(O \hat{\rho}) \right] - \tr(O \rho)} 
 &\le
    \frac{2^n+1}{2^n}
    \left|  \tr \left( \sum_{x,z} \left( \Ex_{\zeta \gets \statedist} [ (\ketbra{\zeta_{x,z}^*})^{\otimes 2} ] - \haarint{\phi}{2} \right) \rho \otimes O \right) \right|
\end{align}
\end{proposition}

\begin{proposition}[Variance Bound]
\label{prop:variance-bound}
Let $\statedist$ be a distribution over pure states, and let $O$ be an observable. For a quantum state $\rho$, let $\hat{\rho}$ be the random variable of the classical shadow with respect to state distribution $\statedist$. Then:

\begin{align}
\label{eq:var-final-bound}
    \text{\emph{Var}}\left[\tr(O \hat{\rho}) \right] 
    &\leq 3 \tr( O_0^2 ) \\
    &\,\,\,\,\, + \Bigg| \frac{(2^n+1)^2}{2^n} \sum_{x,z} \bigg[ \tr \left( \left( \Ex_{\zeta \gets \statedist} [ (\ketbra{\zeta_{x,z}^*})^{\otimes 3} ] - \int \ketbra{\phi}^{\otimes 3} d\mu(\phi) \right) \rho \otimes O \otimes O \right) \bigg] \Bigg|  \nonumber \\
   &\,\,\,\,\, +\Bigg|\frac{2 (2^n+1)^2\tr(O)}{2^{2n}} \sum_{x,z}\bigg[ \tr \left( \left( \Ex_{\zeta \gets \statedist} [ (\ketbra{\zeta_{x,z}^*})^{\otimes 2} ] - \int \ketbra{\phi}^{\otimes 2} d\mu(\phi) \right) \rho \otimes O \right) \bigg] \Bigg| \nonumber
\end{align}
\end{proposition}

In order to prove the above propositions we need the following technical proposition:
\begin{proposition}
\label{prop:probability}
Let $\zeta$ be a unit vector corresponding to a pure state $\ket{\zeta}$. Let $x,z$ be the measurement outcomes when running the shadow tomography process described in figure \ref{fig:design-circuit} with input state $\rho$ and auxiliary state $\ket{\zeta}$. Then:
\begin{align}
    \Pr \Big[x,z \Big| \zeta \Big] = \frac{1}{2^n} \bra{\zeta_{x,z}^*}\rho \ket{\zeta_{x,z}^*} \enspace.
\end{align}
\end{proposition}

\begin{proof}[Proof of Proposition~\ref{prop:probability}]
Direct calculation:
\begin{align}
    \Pr \Big[x,z \Big| \zeta \Big] &= \bra{z,x} (H^{\otimes n} \otimes I) (\mathrm{CNOT}) (\rho \otimes \ketbra{\zeta}) (\mathrm{CNOT}^\dagger) (H^{\dagger \otimes n} \otimes I) \ket{z,x}
    \\ &= 
    \sum_{z', z''} \frac{1}{2^n} (-1)^{z\cdot z' + z\cdot z''}
    \bra{z',x} (\mathrm{CNOT}) (\rho \otimes \ketbra{\zeta}) (\mathrm{CNOT}^\dagger) \otimes I) \ket{z'',x}
    \\ &= 
    \sum_{z', z''} \frac{1}{2^n} (-1)^{z\cdot (z' \oplus z'')}
    \bra{z',x \oplus z'} (\rho \otimes \ketbra{\zeta}) \ket{z'',x \oplus z''}
    \\ &= 
    \sum_{z', z''} \frac{1}{2^n} (-1)^{z\cdot (z' \oplus z'')}
    \bra{z',x \oplus z'} (\rho \otimes Z^z X^x\ketbra{\zeta_{x,z}}X^x Z^z) \ket{z'',x \oplus z''}
    \\ &= 
    \sum_{z', z''} \frac{1}{2^n} (-1)^{z\cdot (z' \oplus z'') + z\cdot(x\oplus z') + z\cdot(x\oplus z'')}
    \bra{z',x \oplus z'} (\rho \otimes X^x \ketbra{\zeta_{x,z}}X^x) \ket{z'',x \oplus z''}
    \\ &= 
    \sum_{z', z''} \frac{1}{2^n} 
    \bra{z', z'} (\rho \otimes \ketbra{\zeta_{x,z}}) \ket{z'', z''}
    \\ &= 
    \frac{1}{2^n} \tr(\rho (\ketbra{\zeta_{x,z}})^T)
    \\ &= 
    \frac{1}{2^n} \tr(\rho \ketbra{\zeta_{x,z}^*})
        \\ &= 
    \frac{1}{2^n} \bra{\zeta_{x,z}^*}\rho \ket{\zeta_{x,z}^*}~,
\end{align}
\end{proof}

\begin{proof}[Proof Of Proposition~\ref{prop:bias-bound}]

By Proposition~\ref{prop:probability}, the linearity and circularity of the trace, and the linearity of the expectation, we have that:
\begin{align}
        \Ex_{\zeta \gets \statedist} \left[\tr(O {\hat{\rho}}) \right]
    &=
    \tr \left( O \cdot \left(\Ex_{\zeta \gets \statedist}  \Ex_{x,z} \left((2^n+1)\ketbra{\zeta_{x,z}^*} - \bbI \right) \right) \right)
    \\&=
    \tr \left( O \cdot \left((2^n+1)\Ex_{\zeta \gets \statedist} \sum_{x,z} \Pr \Big[x,z \Big| \zeta \Big] \ketbra{\zeta_{x,z}^*} - \bbI \right)  \right)
    \\&=
    \tr \left( O \cdot \left((2^n+1) \Ex_{\zeta \gets \statedist} \sum_{x,z} \ketbra{\zeta_{x,z}^*} \cdot \frac{1}{2^n}\bra{\zeta_{x,z}^*} \rho \ket{\zeta_{x,z}^*} - \bbI \right)  \right)
    \\&=
    \tr \left( (2^n+1) O \cdot  \sum_{x,z}  \Ex_{\zeta \gets \statedist} \ketbra{\zeta_{x,z}^*} \cdot \frac{1}{2^n}\bra{\zeta_{x,z}^*} \rho \ket{\zeta_{x,z}^*} -O \cdot \bbI \right)
    \\&=
     \label{eq:expectation-trace}
    \frac{(2^n+1)}{2^n} \tr \left( \sum_{x,z}  \Ex_{\zeta \gets \statedist} \ketbra{\zeta_{x,z}^*}^{\otimes 2} \cdot \rho \otimes O \right) - \tr(O) ~.
\end{align}

Now, by Proposition~\ref{prop:haar_two}, we have that:
\begin{equation}
\label{eq:trace-haar}
 2^n (2^n+1) \tr \left( \left( \haarint{\phi}{2} \right) \rho \otimes O \right) = \tr(\rho) \tr(O) + \tr(O\rho)
\end{equation}
And since $\tr(\rho) = 1$, by summing over $x,z \in \zo^n$, we obtain that:

\begin{align}
\label{eq:expectation-trace-2}
    \tr(O\rho)
    =
    \frac{(2^n+1)}{2^n}  \tr \left( \sum_{x,z}  \haarint{\phi}{2} \cdot (\rho \otimes O) \right) - \tr(O) ~,
\end{align}

The proposition then follows by substituting the terms in the left hand side of equation~\ref{eq:prop_exp} according to equations~\ref{eq:expectation-trace} and \ref{eq:expectation-trace-2}, the linearity of the trace, and the triangle inequality.
\end{proof}

\begin{proof}[Proof of Proposition~\ref{prop:variance-bound}]
Let $O_0 = O-\frac{\tr(O)}{2^n}\bbI$ be the traceless part of $O$. Note the since $O$ and $O_0$ differ only by an identity component, the variance of estimating $O$ and $O_0$ is the same. For any matrix $A$ and scalar $\alpha$, it holds that

\begin{align}
    \tr((\alpha A-\bbI)O_0) = \tr(\alpha A \cdot O_0)
\end{align}

Therefore,
\begin{align}
	\text{Var} \left[ \tr (O\hat{\rho}) \right] &= \text{Var} \left[ \tr (O_0\hat{\rho}) \right] \\
 &= \Ex \left[ \tr (O_0\hat{\rho})^2\right] - \Ex \left[ \tr (O_0\hat{\rho})\right]^2 \\
 & \le \Ex \left[ \tr (O_0\hat{\rho})^2\right] \\
 &=  \Ex_{\zeta \leftarrow \statedist} \left[ \sum_{x,z} \Pr \Big[x,z \Big| \zeta \Big] \tr \left(O_0 \cdot \left((2^n+1)\ketbra{\zeta_{x,z}^*} - \bbI \right) \right)^2 \right] \\
 &=  \Ex_{\zeta \leftarrow \statedist} \left[ \sum_{x,z} \frac{1}{2^n} \bra{\zeta_{x,z}^*} \rho \ket{\zeta_{x,z}^*} \tr \left(O_0 \cdot \left((2^n+1)\ketbra{\zeta_{x,z}^*} \right) \right)^2 \right] \\
 \label{eq:var-eq-1}
	& = \frac{(2^n+1)^2}{2^n}  \tr \left( \Ex_{\zeta \gets \statedist} \left[ \sum_{x,z}  (\ketbra{\zeta_{x,z}^*}))^{\otimes 3} \right] \rho \otimes O_0 \otimes O_0 \right) \enspace .
\end{align}

Substituting back $O_0 = O-\frac{\tr(O)}{2^n}\bbI$ we get
\begin{align}
\label{eq:var-eq-2}
 = \frac{(2^n+1)^2}{2^n} \Bigg[ \tr \left( \Ex_{\zeta \gets \statedist} \left[ \sum_{x,z} (\ketbra{\zeta_{x,z}^*})^{\otimes 3} \right] \rho \otimes O \otimes O \right)  \\   \,\,\,\,\,\,\,\,\,\,\,\,\,\,\,\,\,\,\,\,\,\,\,\,\,\,\,\,\,\,\,\,\,\,\,\,\,\, -   \frac{2 \tr(O)}{2^n}  \tr \left( \Ex_{\zeta \gets \statedist} \left[ \sum_{x,z} (\ketbra{\zeta_{x,z}^*})^{\otimes 2} \right] \rho \otimes O \right) \nonumber \\ 
\,\,\,\,\,\,\,\,\,\,\,\,\,\,\,\,\,\,\,\,\,\,\,\,\,\,\,\,\,\,\,\,\,\,\,\,\,\,
   + \frac{\tr(O)^2}{2^{2n}}  \tr \left( \Ex_{\zeta \gets \statedist} \left[ \sum_{x,z} (\ketbra{\zeta_{x,z}^*}) \right] \rho \right) \Bigg] ~.  \nonumber
\end{align}

Now, note that when $\statedist$ is the Haar measure over the set of $n$ qubits, which is invariant under rotations, we have that:
\begin{align}
\frac{(2^n+1)^2}{2^n}  \tr \left( \Ex_{\zeta \gets \statedist} \left[ \sum_{x,z}  (\ketbra{\zeta_{x,z}^*}))^{\otimes 3} \right] \rho \otimes O_0 \otimes O_0 \right) \\
= 2^n(2^n+1)^2 \tr \left( \haarint{\phi}{3} (\rho \otimes O_0 \otimes O_0 ) \right) \nonumber
\end{align}
In turn, it was shown in \cite{huang2020predicting}[Proposition 1] that 
\begin{align}
\label{eq:var-eq-3}
    2^n(2^n+1)^2  \tr \left( \haarint{\phi}{3} (\rho \otimes O_0 \otimes O_0) \right) 
    \leq
    3 \tr( O_0^2 )
\end{align}
Next, for any state $\ketbra{\zeta^*}$ on $n$ qubits we have that 
\begin{align}
    \sum_{x,z} \ketbra{\zeta_{x,z}^*} = \sum_{x,z} X^xZ^z \ketbra{\zeta^*} Z^zX^x = 2^n\bbI \enspace,
\end{align}
so the last term is equal for any distribution over states $\statedist$ \footnote{A random Pauli twirl is a 1-unitary design.}. Finally, Equation \ref{eq:var-final-bound} follows by adding and subtracting
\begin{align}
    2^n(2^n+1)^2 \tr \left( \int \ketbra{\phi}^{\otimes 3} d\mu(\phi) ( \rho \otimes O_0 \otimes O_0 ) \right)
\end{align}
from equation \ref{eq:var-eq-2} and using the bound in equation \ref{eq:var-eq-3}. The final result is achieved by the triangle inequality.

\end{proof}

\subsection{Relative Approximate Designs}

\begin{theorem}s
\label{theorem:relative_snapshot}
Let $\{\ket{\psi_k}\}_{k \sim \cK}$ be an $\epsilon$-approximate \emph{relative} state 3-design.
Then, any \emph{positive} observable $O$ over $n \geq 2$ qubits can be estimated using a single classical snapshot with bias
\begin{equation}
	\left| \Ex \left[\tr( O \hat{\rho} ) \right] - \tr( O \rho ) \right| \leq 2 \varepsilon \tr(O),
\end{equation}
and variance
\begin{equation}
	\text{\emph{Var}} \left[ \tr ( O \hat{\rho} ) \right] \leq 3\tr( O_0^2 ) + 10 \varepsilon \tr( O )^2 ~.
\end{equation}
\end{theorem}

\begin{proof} 
We bound the bias of the estimator as follows, 
\begin{align}
| \Ex \left[ \tr( O  \hat{\rho} ) \right] - \tr( O \rho ) | 
    & \leq \frac{2^n+1}{2^n}
    \left|  \tr \left( \sum_{x,z} \left( \Ex_{k \gets \cK} [ (\ketbra{\psi_{k,x,z}^*})^{\otimes 2} ] - \haarint{\phi}{2} \right) \rho \otimes O \right) \right| \\
    & \leq \frac{\varepsilon(2^n+1)}{2^n} \cdot \tr \left( \left( \sum_{x,z} \int \ketbra{\phi}^{\otimes 2} d\mu(\phi) \right)  \rho \otimes O \right) \\
      & = \varepsilon 2^n (2^n+1) \cdot \tr \left( \left( \int \ketbra{\phi}^{\otimes 2} d\mu(\phi) \right)  \rho \otimes O \right) \\
     & = \varepsilon \left( \tr(\rho) \tr(O) + \tr( O \rho ) \right) \\
     & \leq 2 \varepsilon \tr(O).
\end{align}

For the first inequality, we use Proposition~\ref{prop:bias-bound}. The second inequality is derived by the definition of a relative $\epsilon$-approximate state 2-design in conjunction to the fact that for all $x,z$, $\{\ket{\psi_{k,x,z}^*}\}_{k \sim \cK}$ is also an approximate 2-design.
We use Proposition~\ref{prop:haar_two}, when going from the third line to the forth line. 

We now bound the variance as follows,
\begin{align}
 \mathrm{Var} \left[ \tr ( \hat{\rho} O ) \right]
 & \le 
 3 \tr( O_0^2 ) +  \varepsilon \, 2^n (2^n+1)^2 \bigg[ \tr \left( \left( \int \ketbra{\phi}^{\otimes 3} d\mu(\phi) \right) \rho \otimes O \otimes O \right) \\
 & \,\,\,\,\,\,\,\,\,\,\,\,\,\,\,\,\,\,\,\,\,\,\,\,\,\,\,\,\,\,\,\,\,\,\,\,\,\,\,\,\,\,\,\,\,\,\,\,\,\,\,\,\,\,\,\,\,\,\,\,\,\,\,\,\,\,\,\,\,\,\,\,\,\,\,\,\,\,\,\,\,\,\,\,\,+ \frac{2 \tr(O)}{2^n} \tr\left( \left( \int \ketbra{\phi}^{\otimes 2} d\mu(\phi) \right) \rho \otimes O \right) \bigg] \nonumber \\
& \leq 3 \tr( O_0^2 ) +  \varepsilon \, 2^n (2^n+1)^2 \Bigg[ \tr( O )^2 \left\lVert  \int \ketbra{\phi}^{\otimes 3} d\mu(\phi) \right\rVert_\infty  \\
 & \,\,\,\,\,\,\,\,\,\,\,\,\,\,\,\,\,\,\,\,\,\,\,\,\,\,\,\,\,\,\,\,\,\,\,\,\,\,\,\,\,\,\,\,\,\,\,\,\,\,\,\,\,\,\,\,\,\,\,\,\,\,\,\,\,\,\,\,\,\,\,\,\,\,\,\,\,\,\,\,\,\,\,\,\,+  \frac{2 \tr(O)^2}{2^n} \left\lVert \int \ketbra{\phi}^{\otimes 2} d\mu(\phi) \right\rVert_\infty \Bigg] \nonumber \\
& = 3 \tr( O_0^2 ) +  \varepsilon 2^n (2^n+1)^2 \left( \frac{6 \tr(O)^2}{2^n(2^n+1)(2^n+2)} + \frac{4 \tr(O)^2}{2^{2n}(2^n+1)} \right) \\
& \leq 3 \tr( O_0^2 ) +  10 \varepsilon \tr(O)^2.
\end{align}
The first inequality follows from proposition~\ref{prop:variance-bound} and the relative approximation. The second inequality follows by bounding each trace via H\"older's inequality, and using $\lVert \rho \rVert_1 = 1$ and $\lVert O \rVert_1 = \tr(O)$, as $O$ is positive.
The equality follows from the well-known fact that $\haarint{\phi}{t} = \Pi_{sym^t} / \dim(sym^t)$ where $sym^t$ is the symmetric subspace over $t$ registers, and in particular $\dim(sym^2)=\frac{2^n(2^n+1)}{2}$ and $\dim(sym^3)=\frac{2^n(2^n+1)(2^n+2)}{6}$. The last inequality holds for $n \geq 2$. The lemma follows.
\end{proof}

\subsection{Additive Approximate Designs}

We now bound the bias and the variance in case the state design is only additive approximation.

\begin{figure}
    \centering
    \begin{mdframed}
    \centering
        \begin{quantikz}
        \lstick{{$R_1$ $\,\,\,\, \rho$}} & \qwbundle{n} &  \ctrl{1} & \gate{H^{\otimes n}} & \meter{z}
        & \setwiretype{c} &  \ctrl{3}
        \\
        \lstick{{$R_2$ $\ket{\zeta}$}} & \qwbundle{n}  & \targ{} & & \meter{x}
        & \setwiretype{c} & & \ctrl{2}
        \\
        \lstick{{$R_3$ $\ket{\zeta}$}} & \qwbundle{n}  & &  & & & \gate{Z^z} & \gate{X^x} & \meter{O^*}
        \\
        \lstick{{$R_4$ $\ket{\zeta}$}} & \qwbundle{n}  & &  & & & \gate{Z^z} & \gate{X^x} & \meter{O^*}
        \end{quantikz}
        
    \textbf{Input}:
        \begin{itemize}
            \item $\ket{\zeta}^{\otimes 3}$ - three copies of a state.
        \end{itemize}
        \textbf{Advice}:
        \begin{itemize}

            \item The value of $\norm{O}_{\infty} = \norm{O^*}_{\infty}$, and
            \item $\rho$ - a quantum state.
        \end{itemize}
        
        \begin{multicols}{2}
            \textbf{Expectation Distinguisher - $\adv_{\Ex}$:}
            \begin{enumerate}
                \item Run the circuit above to get a real-valued outcome (of measuring $O^*$) $\alpha$ in $R_3$.
                \item Calculate $p = \frac{1}{2} + \frac{\alpha}{2\norm{O}_{\infty}}$.
                
                \item Return $1$ with probability $p$ and $0$ with probability $1-p$.
            \end{enumerate}
            
            \columnbreak
            \textbf{Variance Distinguisher - $\adv_{Var}$:}
            \begin{enumerate}
                \item Run the circuit above to get a real-valued outcome (of measuring $O^*$) $\alpha_1$ in $R_3$ and $\alpha_2$ in $R_4$.
                \item Calculate $p = \frac{1}{2} + \frac{\alpha_1\alpha_2}{2\norm{O}^2_{\infty}}$.
                \item Return $1$ with probability $p$ and $0$ with probability $1-p$.
        \end{enumerate}
        \end{multicols}
    \end{mdframed}
    \caption{Description of the distinguishers.}
    \label{fig:distinguishers}
\end{figure}

\begin{lemma}
\label{lemma:additive-reduction}
Let $\statedist$ be a distribution over pure states. Let $\adv_{\Ex}$ be the expectation distinguisher from figure \ref{fig:distinguishers}, and let $\adv_{Var}$ be the variance distinguisher from figure \ref{fig:distinguishers}. Suppose that:
\begin{align}
 & \abs{\Pr_{\zeta \gets \haardist}[1 \gets \adv_{\Ex}(\ketbra{\zeta}^{\otimes 2})] -  \Pr_{\zeta \gets \statedist}[1 \gets \adv_{\Ex}(\ketbra{\zeta}^{\otimes 2}) \big]} \leq \epsilon
\end{align}
And also that
\begin{align}
 & \abs{\Pr_{\zeta \gets \haardist}[1 \gets \adv_{Var}(\ketbra{\zeta}^{\otimes 3})] -  \Pr_{\zeta \gets \statedist}[1 \gets \adv_{Var}(\ketbra{\zeta}^{\otimes 3}) \big]} \leq \epsilon
\end{align}
Then, when running the classical shadow algorithm from figure \ref{fig:design-circuit} with respect to $\statedist$, any \emph{positive} observable $O$ can be estimated using a single classical snapshot with bias
\begin{equation}
	\left| \Ex \left[\tr( O \hat{\rho} )\right] - \tr( O \rho ) \right| \leq 2(2^n+1)\varepsilon \norm{O}_{\infty},
\end{equation}
and variance
\begin{equation}
    \text{\emph{Var}}\left[\tr(O \hat{\rho}) \right] \leq 3 \tr( O_0^2 ) +  6\epsilon\norm{O}^2_\infty(2^n+1)^2
\end{equation}
\end{lemma}

\begin{proof}
We consider the algorithms depicted in figure~\ref{fig:distinguishers}, and use the expectation distinguisher to bound the expectation, and the variance distinguisher to bound the variance.

We start by bounding the expectation. We calculate the probability to get $\alpha$ in $R_3$. We already calculated the probability to get $x,z$ in $R_1,R_2$ in Proposition~\ref{prop:probability}. Recall that $O = \sum_i \alpha_i\Pi_i$, where $\alpha$ are real-valued, so $O^* = \sum_i \alpha_i\Pi^*_i$. Therefore, we get:
\begin{align}
    \Pr \Big[\mathrm{measure\ }\alpha_i \Big| \zeta, \rho \Big] 
    &=
    \sum_{x,z}
    \Pr \Big[x,z \Big| \zeta, \rho \Big]
    \cdot \left| \Pi^*_i X^x Z^z \ket{\zeta} \right|^2 
    \\&=
    \sum_{x,z}
    \frac{1}{2^n} \tr(\rho \ketbra{\zeta_{x,z}^*}) \tr(\Pi^*_i \ketbra{\zeta_{x,z}})
    \\&=
    \sum_{x,z}
    \frac{1}{2^n} \tr(\rho \ketbra{\zeta_{x,z}^*}) \tr(\Pi_i \ketbra{\zeta_{x,z}^*})
    \\&=
    \sum_{x,z}
    \frac{1}{2^n} \tr \left(\ketbra{\zeta_{x,z}^*}^{\otimes 2} \cdot \rho \otimes \Pi_i \right)~,
\end{align}
Where going from the first line to the second line relies on the fact that $\Pi^*_\alpha$ is also an orthogonal projection, so $(\Pi^*_\alpha)^\dagger \Pi^*_\alpha = \Pi^*_\alpha$. Going from the second line to the third line relies on the fact that $\tr(\Pi^*_{\alpha} \ketbra{\zeta_{x,z}}) = \tr(\Pi_{\alpha} \ketbra{\zeta_{x,z}^*})$ since it is a real number, hence equals its conjugate. So, the probability of $\adv_{\Ex}$ to output $1$ is therefore:
\begin{align}
    \Pr[1 \gets \adv_{\Ex}| \zeta, \rho] 
    &= \sum_{i} \left(\frac{1}{2} +\frac{\alpha_i}{2\norm{O}_\infty} \right) \Pr[\alpha_i] \\
    &= \frac{1}{2} + \sum_i \frac{\alpha_i}{2\norm{O}_\infty} \sum_{x,z}
    \frac{1}{2^n} \tr \left(\ketbra{\zeta_{x,z}^*}^{\otimes 2} \cdot \rho \otimes \Pi_i \right) \\
    &= \frac{1}{2} + \frac{1}{2\norm{O}_\infty}\sum_{x,z}
    \frac{1}{2^n} \tr \left(\ketbra{\zeta_{x,z}^*}^{\otimes 2} \cdot \rho \otimes O \right)
\end{align}

We can use the bound on the the distinguishing probability as follows,
\begin{align}
    \epsilon & \ge \abs{\Pr_{\zeta \gets \haardist}[1 \gets \adv_{\Ex}(\ketbra{\zeta}^{\otimes 2})] -  \Pr_{\zeta \gets \statedist}[1 \gets \adv_{\Ex}(\ketbra{\zeta}^{\otimes 2}) \big]}
    \\& =
    \frac{1}{2\norm{O}_{\infty}}\left| \tr \left( \frac{1}{2^n} \sum_{x,z}\left( \haarint{\phi}{2} - \Ex_{\zeta \gets \statedist} \left[ \ketbra{\zeta_{x,z}^*}^{\otimes 2} \right]\right) \cdot \rho \otimes O \right) \right|
    \label{eq:2-design-bound}
    \\&\geq
    \frac{1}{2\norm{O}_{\infty}(2^n+1)}\left| \tr(O\rho) - \Ex\left[\tr(O\hat{\rho})\right] \right| ~,
\end{align}
Where the last inequality follows from Proposition~\ref{prop:bias-bound}. This proves the bound on the expectation.
We turn to bound the variance. Recall that by Proposition~ \ref{prop:variance-bound}, we have that
\begin{align} \label{eq:bounds-on-var}
    \text{\emph{Var}}\left[\tr(O \hat{\rho}) \right] 
    &\leq 3 \tr( O_0^2 ) \\
    &\,\,\,\,\, + \Bigg| \frac{(2^n+1)^2}{2^n} \sum_{x,z} \bigg[ \tr \left( \left( \Ex_{\zeta \gets \statedist} [ (\ketbra{\zeta_{x,z}^*})^{\otimes 3} ] - \int \ketbra{\phi}^{\otimes 3} d\mu(\phi) \right) \rho \otimes O \otimes O \right) \bigg] \Bigg|  \nonumber \\
   &\,\,\,\,\, +\Bigg|\frac{2 (2^n+1)^2\tr(O)}{2^{2n}} \sum_{x,z}\bigg[ \tr \left( \left( \Ex_{\zeta \gets \statedist} [ (\ketbra{\zeta_{x,z}^*})^{\otimes 2} ] - \int \ketbra{\phi}^{\otimes 2} d\mu(\phi) \right) \rho \otimes O \right) \bigg] \Bigg| \nonumber \enspace .
\end{align}
We turn to bound the terms in above expression. First, similarly as with the expectation bound, we can calculate the probability of measuring $\alpha_i, \alpha_j$:
\begin{align}
    \Pr \Big[\mathrm{measure\ }\alpha_i,\alpha_j \Big| \zeta, \rho \Big] 
    =
    \sum_{x,z}
    \frac{1}{2^n} \tr \left(\ketbra{\zeta_{x,z}^*}^{\otimes 3} \cdot \rho \otimes \Pi_{i} \otimes \Pi_{j} \right)~.
\end{align}
Therefore, by similar argument, the probability of the distinguisher $\adv_{Var}$ to output 1 is
\begin{align}
    \Pr[1 \gets \adv_{Var} | \zeta, \rho] 
    &= \sum_{i,j}(\frac{1}{2} + \frac{\alpha_i\alpha_j}{2\norm{O}^2_{\infty}}) \cdot \Pr[\alpha_i,\alpha_j] \\
    &= \frac{1}{2} + \sum_{i,j}\frac{\alpha_i\alpha_j}{2\norm{O}^2_{\infty}} \cdot \sum_{x,z}
    \frac{1}{2^n} \tr \left(\ketbra{\zeta_{x,z}^*}^{\otimes 3} \cdot \rho \otimes \Pi_{i} \otimes \Pi_{j} \right) \\
    &= \frac{1}{2} + \frac{1}{2\norm{O}_{\infty}^2} \cdot \sum_{x,z}
    \frac{1}{2^n} \tr \left(\ketbra{\zeta_{x,z}^*}^{\otimes 3} \cdot \rho \otimes O \otimes O \right) ~.
\end{align}
So, by applying to bound on the distinguishing probability we get that
\begin{align}
    \epsilon & \ge \abs{\Pr_{\zeta \gets \haardist}[1 \gets \adv_{Var}(\ketbra{\zeta}^{\otimes 3})] -  \Pr_{\zeta \gets \statedist}[1 \gets \adv_{Var}(\ketbra{\zeta}^{\otimes 3}) \big]}
    \\& =
    \frac{1}{2\norm{O}^2_{\infty}}\left| \tr  \left(\frac{1}{2^n}\sum_{x,z}\left(\haarint{\phi}{3} -  
    \Ex_{\zeta \gets \statedist} \left[ \ketbra{\zeta_{x,z}^*}^{\otimes 3}\right]\right) \cdot \rho \otimes O \otimes O\right) \right| \enspace .
\end{align}
Or in other words
\begin{align} \label{eq:var-bound-3-terms}
    \Bigg| \frac{(2^n+1)^2}{2^n} \sum_{x,z} \bigg[ \tr \left( \left( \Ex_{\zeta \gets \statedist} [ (\ketbra{\zeta_{x,z}^*})^{\otimes 3} ] - \int \ketbra{\phi}^{\otimes 3} d\mu(\phi) \right) \rho \otimes O \otimes O \right) \bigg] \Bigg| \leq 2 \epsilon \norm{O}^2_{\infty} (2^n+1)^2 \enspace .
\end{align}
We turn to bound the second expression in Inequality~\ref{eq:bounds-on-var}. By Inequality~\ref{eq:2-design-bound} we have that
\begin{align} \label{eq:var-bound-2-terms}
    \Bigg|\frac{2 (2^n+1)^2\tr(O)}{2^{2n}} \sum_{x,z} \tr \left( \left( \Ex_{k \gets \cK} [ (\ketbra{\psi_{k,x,z}^*})^{\otimes 2} ] - \int \ketbra{\phi}^{\otimes 2} d\mu(\phi) \right) \rho \otimes O \right) \Bigg| &\leq \\
    & \leq \frac{4\epsilon \abs{\tr(O)}\norm{O}_\infty(2^n+1)^2}{2^n} \\
    &\leq 4\epsilon \norm{O}^2_\infty(2^n+1)^2 \enspace,
\end{align}
where the last inequality follows since $\abs{\tr(O)} \leq 2^n\norm{O}_\infty$. The final bound is achieved by combining Inequality~\ref{eq:var-bound-3-terms} and Inequality~\ref{eq:var-bound-2-terms}.
\end{proof}

\begin{theorem}
\label{theorem:addative_snapshot}
Let $\{\ket{\psi_k}\}_{k \sim \cK}$ be an $\epsilon$-approximate \emph{additive} state 3-design. Then, any \emph{positive} observable $O$ can be estimated using a single classical snapshot with bias

\begin{equation}
	\left| \Ex \left[\tr( O \hat{\rho} )\right] - \tr( O \rho ) \right| \leq (2^n+1)\varepsilon \norm{O}_{\infty},
\end{equation}
and variance

\begin{equation}
    \text{\emph{Var}}\left[\tr(O \hat{\rho}) \right] \leq 3 \tr( O_0^2 ) + 3\epsilon\norm{O}^2_\infty(2^n+1)^2
\end{equation}
\end{theorem}

\begin{proof}
Direct application of lemma \ref{lemma:additive-reduction}. By the operational interpretation of the $\ell_1$-norm for density matrices, it holds that
\begin{multline}
       \Bigg| \Pr_{k \gets \cK}[1 \gets \adv_{\Ex}(\ketbra{\psi_k}^{\otimes 2})] - \Pr_{\phi \gets \mu}[1 \gets \adv_{\Ex}(\ketbra{\psi}^{\otimes 2})] \Bigg| \\
    \leq \frac{1}{2} \norm{\Ex_{k \gets \cK} [ (\ketbra{\psi_k})^{\otimes 2} ] - \int \ketbra{\phi}^{\otimes 2} d\mu(\phi)}_1 
   \leq \frac{\varepsilon}{2} \enspace,
\end{multline}
and similarly it holds that
\begin{multline}
   \Bigg| \Pr_{k \gets \cK}[1 \gets \adv_{Var}(\ketbra{\psi_k}^{\otimes 3})] - \Pr_{\phi \gets \mu}[1 \gets \adv_{Var}(\ketbra{\psi}^{\otimes 3})] \Bigg|  \\
    \leq \frac{1}{2} \norm{\Ex_{k \gets \cK} [ (\ketbra{\psi_k})^{\otimes 3} ] - \int \ketbra{\phi}^{\otimes 3} d\mu(\phi)}_1
   \leq \frac{\varepsilon}{2} \enspace .
\end{multline}
The result then follows from lemma \ref{lemma:additive-reduction}.
\end{proof}

\subsection{State Pseudo-Designs}

Next, we  consider a state design that is only \emph{compuationally} close to a Haar random state, and show that it suffices to estimate any efficient observable.

For simplicity, we allow real-valued classical wires with arbitrary polynomial many digits as inputs or outputs of a quantum circuit. A quantum circuit can toss a coin with probability $p$ for a real-valued $p$  with arbitrary polynomial many digits. We also assume that the gate set is closed under complex conjugation.

\begin{theorem}
\label{theorem:computational_snapshot}
Let $O$ be observable over $n$ qubits with complexity $t$. Let $(\keygen,\stategen)$ be a $(T,\epsilon)$-State $3$-pseudo-design Generator. For a quantum state $\rho$, let $\hat{\rho}$ be the random variable of the classical snapshot with respect to state distribution induced by $(\keygen,\stategen)$. Then there exists a constant $c$ such that if  $T = c \cdot \max(t, n)$, $O$ can be estimated using a single classical snapshot with bias
\begin{equation}
	\left| \Ex \left[\tr( O \hat{\rho} )\right] - \tr( O \rho ) \right| \leq 2(2^n+1)\varepsilon \norm{O}_{\infty},
\end{equation}
and variance
\begin{equation}
 \text{\emph{Var}}\left[\tr(O \hat{\rho}) \right] \leq 3 \tr( O_0^2 ) +  6\epsilon\norm{O}^2_\infty(2^n+1)^2
\end{equation}
\end{theorem}
\begin{proof}

This is a direct application of Lemma~\ref{lemma:additive-reduction}. In particular, if $O$ has complexity $t$, so does $O^*$. Flipping a coin in the reduction with the desired probability can be done by a quantum circuit of size $O(t)$. So overall, the reductions $\adv_{\Ex}$ and $\adv_{Var}$ have complexity $O(\max(t,n))$. Therefore, there exists a constant $c$ such that by setting $T = c \cdot \max(t, n)$, we can apply to the security of the $3$-pseudo-design Generator, and by  Lemma~\ref{lemma:additive-reduction} get the desired result.
\end{proof}

\subsection{Proof of main theorems}
\label{section:proofs-main-theorem}
We now ready to prove our main theorems.
\begin{proof}[Proof of Theorems \ref{theorem:relative-guarantees}, \ref{theorem:additive-guarantees} and \ref{theorem:computational-guarantees}]
We begin with Theorem~\ref{theorem:relative-guarantees}. The claim follows from combining the variance and bias bounds from Theorem~\ref{theorem:relative_snapshot} with rigorous performance guarantee for median of means estimation \cite{nemirovski1983medianmeans,JERRUM1986169}: Let $X$ be a random variable with variance $\sigma^2$. Then, $K$ independent sample means of size $L=34 \sigma^2/\gamma^2$ suffice to construct a median of means estimator $\hat{\mu}(L,K)$ that obeys
$\mathrm{Pr} \left[ \left| \hat{\mu}(L,K)-\mathbb{E} \left[ X \right] \right| \geq \gamma \right] \leq 2 \mathrm{e}^{-K/2}$ for all $\gamma >0$.
The parameters $L$ and $K$ are chosen such that this general statement ensures
$
\mathrm{Pr} \left[ \left| \hat{o} (L,K) - \Ex\left(\tr\left( O_i \hat{\rho} \right)\right) \right| \geq \epsilon \right] \leq \delta
$.
Finally, by Theorem~\ref{theorem:relative_snapshot} we have the following bound on the bias:
\begin{align}
    \left| \Ex \left[\tr( O \hat{\rho} ) \right] - \tr( O \rho ) \right| \leq 2 \varepsilon \tr(O) \enspace .
\end{align}
The result is then obtained by the triangle inequality.

The proofs of Theorem~\ref{theorem:additive-guarantees} and Theorem~\ref{theorem:computational-guarantees} are by the same argument, by plugging in the corresponding bounds on the variance and bias from Theorem~\ref{theorem:addative_snapshot} and Theorem~\ref{theorem:computational_snapshot}, respectively.
\end{proof}

\appendix

\section{Additive vs.\ Relative Approximate State Designs}
\label{apx:addrelstate}

In this section we show for completeness how additive approximation translates to relative approximation. For $n, t \in \mathbb{N}$, denote $N = 2^n$, and let $Sym_{N,t}$ be the symmetric subspace over $n$ qubits with $t$ copies. We begin with showing a positive result, which resembles the results of \cite{schuster2024random, Brand_o_2016} for unitary designs:
\begin{lemma}
\label{lemma:additive-to-relative}
Let $n, t, \lambda\in \mathbb{N}$ and let $\epsilon >0$. Let $\cK$ be a distribution over $\binset^\secp$ such that the ensemble of states $\{ \ket{\psi_k} \}_{k \sim \cK}$ is an \emph{additive} $\epsilon$-approximate state $t$-design over $n$ qubits. Then $\{ \ket{\psi_k} \}_{k \sim \cK}$ is also a \emph{relative} $\epsilon'$-approximate state $t$-design, for $\epsilon' \geq dim(Sym_{N,t}) \cdot \epsilon$ \footnote{Note it could be the case that $\epsilon' \geq 1$, if $\epsilon \geq \frac{1}{dim(Sym_{N,t})}$.}.
\end{lemma}

\begin{proof} 
Let $\Pi_{Sym_{N,t}}$ be the projection onto the symmetric subspace $Sym_{N,t}$. Let
\begin{align}
   M = \mathbb{E}_{k \gets \cK}\left[\ketbra{\psi_k}^{\otimes t}\right] \enspace ,
\end{align}
and similarly let
\begin{align}
  H = \haarint{\phi}{t}  \enspace .
\end{align}
We will show that $M \preceq (1+dim(Sym_{N,t}) \cdot \epsilon) \cdot H$, the other direction is symmetrical.
Note that by construction $\Pi_{Sym_{N,t}} M = M$. Moreover, it is well known that $H = \frac{1}{dim(Sym_{N,t})}\Pi_{Sym_{N,t}}$ \cite{harrow2013churchsymmetricsubspace}. It follows that in an orthogonal diagonalized form we can write
\begin{align}
    M - H = \sum_{i} \alpha_i \ketbra{\psi_i} \enspace ,
\end{align}
Where for every $i$, $\psi_i \in Sym_{N,t}$. So,
\begin{align}
    M - H \preceq \max_i{\abs{\alpha_i}} \cdot \Pi_{Sym_{N,t}} \enspace .
\end{align}
Next, by our assumption it follows that
\begin{align}
    \norm{M - H}_\infty \leq \norm{M - H}_1 \leq \epsilon \enspace .
\end{align}
Combined, we get that
\begin{align}
     M - H \preceq \epsilon \cdot \Pi_{Sym_{N,t}} \enspace ,
\end{align}
which implies one direction, since $H = \frac{1}{dim(Sym_{N,t})}\Pi_{Sym_{N,t}}$. The other direction is symmetrical.
\end{proof}

\begin{corollary}
\label{corol:additive-to-relative}
Let $n, t, \lambda\in \mathbb{N}$ and let $\epsilon >0$. Let $\cK$ be a distribution over $\binset^\secp$ such that the ensemble of states $\{ \ket{\psi_k} \}_{k \sim \cK}$ is an \emph{additive} $\epsilon$-approximate state $t$-design over $n$ qubits. Then $\{ \ket{\psi_k} \}_{k \sim \cK}$ is also a \emph{relative} $\epsilon'$-approximate state $t$-design, for $\epsilon' = 2^{n \cdot t} \cdot \epsilon$.
\end{corollary}

\begin{proof}
This follows from Lemma~\ref{corol:additive-to-relative} since $dim(Sym_{N,t}) \leq N^t$.
\end{proof}

Next we show that in some cases Lemma~\ref{lemma:additive-to-relative} is tight up to a constant factor.

\begin{lemma}
\label{lemma:lower-bound}
Let $n, t \in \mathbb{N}$ such that $n \geq 2$, and let $0 < \epsilon < 1$. Then there exists an  ensemble of states $\{ \ket{\psi_k} \}_{k \sim \cK}$ which is an \emph{additive} $\epsilon$-approximate state $t$-design over $n$ qubits and not a \emph{relative} $\epsilon'$-approximate state $t$-design, for $\epsilon' < \frac{dim(Sym_{N,t}) \cdot \epsilon}{3}$.
\end{lemma}

\begin{proof}
Let 
\begin{align}
    H = \haarint{\phi}{t}  \enspace ,
\end{align}
be $t$ copies of the Haar random distribution over $n$ qubits. Let $N = 2^n$ and $\psi \in \mathbb{C}^{N}$ be some unit vector, and let
\begin{align}
    M = (1-\frac{\epsilon}{2}) H + \frac{\epsilon}{2} \ketbra{\psi}^{\otimes t} \enspace .
\end{align}
It follows that
\begin{align}
    \norm{M - H}_{1} \leq \epsilon
\end{align}
Consider now the projection $\Pi_{\psi} = \ketbra{\psi}^{\otimes t}$. Let $\Pi_{Sym_{N,t}}$ be the projection onto the symmetric subspace $Sym_{N,t}$. Next, observe that by construction $\Pi_{\psi}$ is a projection on a subspace of the symmetric subspace, and in particular, $\Pi_{\psi}\Pi_{Sym_{N,t}} = \Pi_{\psi}$. Recall that $H = \frac{1}{dim(Sym_{N,t})}\Pi_{Sym_{N,t}}$ \cite{harrow2013churchsymmetricsubspace}. Therefore, we have that
\begin{align}
    \tr(\Pi_{\psi}H) &= \\
    &= \tr(\Pi_{\psi}\frac{1}{dim(Sym_{N,t})}\Pi_{Sym_{N,t}}) \\
    &= \frac{1}{dim(Sym_{N,t})}\tr(\Pi_{\psi}) \\
    &= \frac{1}{dim(Sym_{N,t})} \enspace .
\end{align}
We also have that
\begin{align}
    \tr(\Pi_{\psi}M) &= \\
    &= \tr(\Pi_{\psi}((1-\frac{\epsilon}{2}) H + \frac{\epsilon}{2} \ketbra{\psi}^{\otimes t}) \\
    &= (1-\frac{\epsilon}{2})\tr(\Pi_{\psi} H) +  \frac{\epsilon}{2}\tr(\Pi_{\psi} \ketbra{\psi}^{\otimes t}) \\
    &= \frac{1-\frac{\epsilon}{2}}{dim(Sym_{N,t})} + \frac{\epsilon}{2}
\end{align}
Now, in order to $M$ be a relative $\epsilon'$-approximate state $t$-design, it must hold that 
\begin{align}
    \tr(\Pi_{\psi}M) \leq (1+\epsilon')\tr(\Pi_{\psi}H) \enspace ,
\end{align}
or it other words it must hold that
\begin{align}
    \frac{1-\frac{\epsilon}{2}}{dim(Sym_{N,t})} + \frac{\epsilon}{2} \leq \frac{1+\epsilon'}{dim(Sym_{N,t})}
\end{align}
Simplifying, we get
\begin{align}
    \frac{\epsilon}{2}(dim(Sym_{N,t}) - 1) \leq \epsilon'
\end{align}
Next, for $dim(Sym_{N,t}) \geq 3$ we have that
\begin{align}
    \frac{dim(Sym_{N,t}) \cdot \epsilon}{3} \leq \frac{\epsilon}{2}(dim(Sym_{N,t}) - 1)
\end{align}
Finally, by monotonicity of $dim(Sym_{N,t})$ in $N$ and $t$, it follows that $dim(Sym_{N,t}) \geq 3$ when $t \geq 1, n \geq 2$.
\end{proof}

\begin{corollary}
\label{corol:lower-bound-additive-to-relative}
Let $n, t \in \mathbb{N}$ such that $n \geq 2$, and let $0 < \epsilon < 1$. Then there exists an  ensemble of states $\{ \ket{\psi_k} \}_{k \sim \cK}$ which is an \emph{additive} $\epsilon$-approximate state $t$-design over $n$ qubits and not a \emph{relative} $\epsilon'$-approximate state $t$-design, for $\epsilon' < \frac{2^{n\cdot t} \cdot \epsilon}{3\cdot t!}$.
\end{corollary}

\begin{proof}
This follows from Lemma~\ref{lemma:lower-bound} since $dim(Sym_{N,t}) \geq \frac{N^t}{t!}$.
\end{proof}

Finally, for completeness, we show that a relative approximation also yields an additive one. 

\begin{lemma}
\label{lemma:relative-to-additive}
Let $n, t,\lambda\in \mathbb{N}$ and let $\epsilon >0$. Let $\cK$ be a distribution over $\binset^\secp$ such that the ensemble of states $\{ \ket{\psi_k} \}_{k \sim \cK}$ is a \emph{relative} $\epsilon$-approximate state $t$-design over $n$ qubits. Then $\{ \ket{\psi_k} \}_{k \sim \cK}$ is also an \emph{additive} $\epsilon$-approximate state $t$-design.
\end{lemma}

\begin{proof}
Let
\begin{align}
   M = \mathbb{E}_{k \gets \cK}\left[\ketbra{\psi_k}^{\otimes t}\right] \enspace ,
\end{align}
and similarly let
\begin{align}
  H = \haarint{\phi}{t}  \enspace .
\end{align}
Let $\Pi_{Sym_{N,t}}$ be the projection onto the symmetric subspace $Sym_{N,t}$, and that by construction $\Pi_{Sym_{N,t}} M = M$. Recall that $H = \frac{1}{dim(Sym_{N,t})}\Pi_{Sym_{N,t}}$ \cite{harrow2013churchsymmetricsubspace}. It follows that in an orthogonal diagonalized form we can decompose
\begin{align}
    M = \sum_i \alpha_i\ketbra{\psi_i} \enspace ,
\end{align}
and similarly, we can decompose \footnote{Note that for every $i$ we have that $\beta_i = \frac{1}{dim(Sym_{N,t})}$.}
\begin{align}
    H = \sum_i \beta_i\ketbra{\psi_i}  \enspace .
\end{align}
Where for every $i$, $\psi_i \in Sym_{N,t}$.
By denoting $\gamma_i = \alpha_i - \beta_i$, we have that 
\begin{align}
    M - H = \sum_i \gamma_i\ketbra{\psi_i} \enspace .
\end{align}
Next, by the assumption that $M$ is a \emph{relative} $\epsilon$-approximate state $t$-design, it follows that
\begin{align}
    -\epsilon H \preceq M - H \preceq \epsilon H \enspace,
\end{align}
so for every $i$ we have that: 
\begin{align}
    \abs{\gamma_i} \leq \epsilon \cdot \beta_i \enspace.
\end{align}
Therefore, by the definition of the $\ell_1$ norm, and since $\sum_i \beta_i = 1$, it follows that
\begin{align}
    \norm{M - H}_1 = \sum_i \abs{\gamma_i} \leq \sum_i \epsilon\beta_i \leq \epsilon \enspace .
\end{align}
\end{proof}

\bibliographystyle{alpha}
\bibliography{Bibliography}

\end{document}